\documentclass[letterpaper, 10 pt, conference]{ieeeconf}  % Comment this line out if you need a4paper

\IEEEoverridecommandlockouts                              % This command is only needed if 
\usepackage{subcaption}
\usepackage{amssymb,amsmath,bm,booktabs}
\usepackage{graphicx}
\usepackage[colorinlistoftodos]{todonotes}
\usepackage[colorlinks=true, allcolors=blue]{hyperref}
\usepackage[authormarkuptext=name,addedmarkup=bf,authormarkupposition=left]{changes}
\definechangesauthor[name={M.~L.}, color={blue}]{ml}
\definechangesauthor[name={S.~O.}, color={red}]{so}
\definechangesauthor[name={K.~S.}, color={green}]{ks}

\newcommand{\sXX}[1]{\added[id=so,remark={}]{#1}}

\newtheorem{theorem}{Theorem}%[section]

\newtheorem{definition}[theorem]{Definition}

\usepackage{caption}
\usepackage{algpseudocode, algorithm}
\usepackage{comment}

\DeclareMathOperator*{\argmax}{arg\,max}

\title{\LARGE \bf Learning for Multi-robot Cooperation in Partially Observable Stochastic Environments with Macro-actions}

\author{Miao Liu$^{1}$, Kavinayan Sivakumar$^{2}$, Shayegan Omidshafiei$^{3}$, Christopher Amato$^{4}$ and Jonathan P. How$^{3}$% <-this % stops a space
\thanks{$^{1}$Miao Liu is with IBM T. J. Watson Research Center, Yorktown Heights, NY, USA {\tt\small miao.liu1@ibm.com}}%
\thanks{$^{2}$Kavinayan Sivakumar is with the Department of Electrical Engineering, Princeton University, Princeton, NJ, USA {\tt\small ks16@princeton.edu}}%
\thanks{$^{3}$ Shayegan Omidshafiei and Jonathan P. How are with Laboratory of Information and Decision Systems, Massachusetts Institute of Technology, Cambridge, MA, USA {\tt\small \{shayegan, jhow\}@mit.edu}}%
\thanks{$^{4}$Christopher Amato is with the College of Computer and Information Science, Northeastern University, Boston, MA, USA {\tt\small camato@ccs.neu.edu}}%
}

\renewcommand{\baselinestretch}{0.975}
\begin{document}

\maketitle
\thispagestyle{empty}
\pagestyle{empty}

\begin{abstract}
This paper presents a data-driven approach for multi-robot coordination in partially-observable domains based on Decentralized Partially Observable Markov Decision Processes (Dec-POMDPs) and macro-actions (MAs). Dec-POMDPs provide a general framework for cooperative sequential decision making under uncertainty and MAs allow temporally extended and asynchronous action execution. To date, most methods assume the underlying Dec-POMDP model is known a priori or a full simulator is available during planning time. Previous methods which aim to address these issues suffer from local optimality and sensitivity to initial conditions. Additionally, few hardware demonstrations involving a large team of heterogeneous robots and with long planning horizons exist. This work addresses these gaps by proposing an iterative sampling based Expectation-Maximization algorithm (iSEM) to learn polices %\sXX{is SMCEM ever mentioned again in the paper?} to learn policies %\todo{CA: It would be good to make it more clear what the contribution is algorithmicly and experimentally}
using only trajectory data containing observations, MAs, and rewards. Our experiments show the algorithm is able to achieve better solution quality than the state-of-the-art learning-based methods. We implement two variants of multi-robot Search and Rescue (SAR) domains (with and without obstacles) on hardware to demonstrate the learned policies can effectively control a team of distributed robots to cooperate in a partially observable stochastic environment.%\sXX{heads up: footer on this page is for AAAI17, change before submission!}
\end{abstract}

\section{Introduction}

There has been significant progress in recent years on developing cooperative multi-robot systems that can operate in real-world environments with uncertainty. 
%With the increasing use of teams of autonomous systems for many applications, robust decentralized methods that coordinate time-sensitive tactical tasks must be developed to achieve mission objectives in dynamic and uncertain environments. 
Example applications of social and economical interest include search and rescue (SAR) \cite{grayson2014search},  traffic management for smart cities \cite{dresner2006multiagent}, planetary navigation \cite{bernstein2001planetary}, robot soccer \cite{jolly2007intelligent}, e-commerce and transport logistic processes~\cite{gath2016optimizing}. %~\cite{AmatoICRA15}\todo{CA: I assume this is going to be updated with more citations and motivation?}.  
%\todo{citations?}
 Planning in such environments must address numerous challenges, including imperfect models and knowledge of the environment, restricted communications between robots, noisy and limited sensors, different viewpoints by each robot, asynchronous calculations, and computational limitations. 
 
These planning problems, in the most general form, can be formulated as a Decentralized Partially Observable Markov Decision Process (Dec-POMDP)~\cite{Book16}, a general framework for cooperative sequential decision making under uncertainty. In Dec-POMDPs, robots make decisions based on local streams of information (i.e., observations), such that the expected value of the team (e.g., number of victims rescued, average customer satisfaction) is maximized. However, representing and solving Dec-POMDPs is often intractable for large domains, because finding the optimal (even $epsilon$-approximate) solution of a Dec-POMDP (even for finite horizon) is NEXP-complete~\cite{Book16}. %\sXX{mention why}. 
To combat this issue, recent research has addressed the more scalable macro-action based Dec-POMDP (MacDec-POMDP), where each agent has temporally-extended actions, which may require different amounts of time to complete~\cite{amato2014planning}. Moreover, significant progress has been made on demonstrating the usefulness of MacDec-POMDPs via a range of challenging robotics problems, such as a warehouse domain~\cite{AmatoICRA15}, bartending and beverage service~\cite{CA:RSS15}, and package delivery~\cite{ShayeganICRA15,omidshafiei2016graph}. However, current MacDec-POMDP methods require knowing domain models a priori. Unfortunately, for many real-world problems, such as SAR, the domain model may not be completely available. Recently, researchers started to address this issue via reinforcement learning and proposed a policy-based EM algorithm (PoEM)~\cite{Liu:AAAI16}, which can learn valid controllers via only trajectory data containing observations, macro-actions (MAs), and rewards.  
%\todo[inline, color=green!40]{highlight the problem}
%\todo[inline, color=green!40]{discuss the development history of decPOMDPs. The importance of off-policy learning: safety.}
 %With the increasing use of teams of autonomous systems for many applications, robust decentralized methods that coordinate time-sensitive tactical tasks must be developed to achieve mission objectives in dynamic and uncertain environments. 
%Autonomous  systems  hardware  has  reached  a  point  where  vehicles  and  robots  are  capable  of  operating in these environments.   However,  existing general methods for controlling a team of robots in uncertain environments are not yet capable enough to execute many missions.

Although PoEM has convergence guarantees for the batch learning setting and can recover optimal policies for benchmark problems with sufficient data, it suffers from local optimality and sensitivity to initial conditions for complicated real-word problems. Inevitably, as an EM type algorithm, the results of PoEM can be arbitrarily poor given bad initialization. % \sXX{this needs to be made stronger. e.g., can the results of poem be arbitrarily poor given bad initialization? also, from fig 2c, it seems high variance in results is an issue in poem, so maybe emphasize that more}. 
Additionally, few hardware demonstrations based on challenging tasks such as SAR, which involves a large team of heterogeneous robots (both ground vehicles and aerial vehicles) and with MacDec-POMDP formulation exists. %\sXX{seems contradictory to prev paragraph, where applications were mentioned and cited. need to emphasize the difficulty of SAR case more.}. 
This paper addresses these gaps by proposing an iterative sampling-based Expectation-Maximization algorithm (iSEM) to learn polices. Specifically, this paper extends previous approaches by using concurrent (multi-threaded) EM iterations providing feedback to one another to enable re-sampling of parameters and reallocation of computational resources for threads which are clearly converging to poor values.% (so computational efforts are not wasted).\mXX{does it sound better now?}%so computational efforts are not wasted on runs of EM which are clearly converging to poor values
%\todo{--CA: this needs to sound cooler. It doesn't sound too much better, but I tried to help a little bit. The point is that we don't want to just sound like we just run a bunch of PoEM instances with different initalizations.}. 

The algorithm is tested in batch learning settings, which is commonly used in learning from demonstration. %Moreover, the  proposed algorithm performs off-policy batch learning to improve the policy with guaranteed convergence.  
Through theoretical analysis and numerical comparisons on a large multi-robot SAR domain, we demonstrate the new algorithm can better explore the policy space. As a result, iSEM is able to achieve better expected values compared to the state-of-the-art learning-based method, PoEM.  Finally, we present an implementation of two variants of multi-robot SAR domains (with and without obstacles) on hardware to demonstrate the learned policies can effectively control a team of distributed robots to cooperate in a partially observable stochastic environment.

\section{Background}

We first discuss the background on Dec-POMDPs and MacDec-POMDPs and then describe the PoEM algorithm. 

\subsection{Dec-POMDPs and MacDec-POMDPs}\label{sec:model-framework}
%\todo{CA: maybe break this up into multiple sections?}
Decentralized POMDPs (Dec-POMDPs) generalize POMDPs to the multiagent, decentralized setting \cite{Book16,Bernstein:MOR-2002}. Multiple agents operate under uncertainty
based on partial views of the world, with execution unfolding over a bounded or unbounded number of steps. At
each step, every agent chooses an action (in parallel) based on locally observable information and then receives a new observation. The agents share a joint reward based on their joint concurrent actions, making the problem cooperative. However, agents' local views mean that execution is decentralized.

Formally, a \textbf{Dec-POMDP} is represented as an octuple $\langle N, A, S, Z, T, \Omega, R, \gamma \rangle$, where $N$ is a finite set of agent indices; $A=\otimes_{n} A_n$ and $Z=\otimes_n Z_n$ respectively are sets of joint actions and observations, with $A_n$ and $Z_n$ available to agent $n$. At each step, a joint action $\vec{a}=(a_1, \cdots, a_{|N|})\in A$ is selected and a joint observation $\vec{z}=(z_1, \cdots, z_{|N|})$ is received; $S$ is a set of finite world states; $T : S\times A\times S\rightarrow [0,1]$ is the state transition function with $T(s'|s,\vec{a})$ denoting the probability of transitioning to $s'$ after taking joint action $\vec{a}$ in $s$; $\Omega : S\times A\times Z\rightarrow [0,1]$ is the observation function with $\Omega(\vec{z}|s', \vec{a})$ the probability of observing $\vec{o}$ after taking joint action $\vec{a}$ and arriving in state $s'$; $R : S\times A\rightarrow{\mathbb{R}}$ is the reward function with $r(s, \vec{a})$ the immediate reward received after taking joint action $\vec{a}$ in $s$; $\gamma\in[0,1)$ is a discount factor. %A global reward signal is generated for the team of agents after joint actions are taken, but each agent only observes its local observation. %(the global rewards are only used during policy learning; they are not needed when executing the policy).  
Because each agent lacks access to other agents' observations, each agent maintains a local policy $\pi_n$, defined as a mapping from local observation histories to actions. A joint policy consists of the local policies of all agents. For an infinite-horizon Dec-POMDP with initial state $s_0$, the objective is to find a joint policy $\pi=\otimes_n \pi_n$, such that the value of  $\pi$ starting from $s_0$, $V^\pi(s_0) = \mathbb{E}\big[\sum_{t=0}^{\infty}\gamma^tr(s_t, \vec{a}_t)|s_0, \pi\big]$, is maximized. Specifically, given $h_t = \{a_{0:t-1},z_{0:t}\}\in H_n$, the history of actions and observations up to $t$, the policy $\pi_n$ probabilistically maps $h_t$ to $a_t$: $H_n\times A_n\rightarrow [0,1]$.

A MacDec-POMDP with (local) macro-actions extends the MDP-based options~\cite{sutton1999between} framework to Dec-POMDPs. Formally, a \textbf{MacDec-POMDP} is defined as a tuple $\langle N,A,M,S,Z,O,T,\Omega,R,\gamma \rangle$, where $N,A,S,Z,T,\Omega,R$ and $\gamma$ are the same as defined in the Dec-POMDP; $O=\otimes O_n$ are sets of joint macro-action observations which are functions of the state; $M=\otimes M_n$ are sets of joint macro-actions, with $M_n\!=\!\langle I_{n}^m,\beta_{n}^m,\pi_{n}^m\rangle$, where $I_n^m\!\subset\!H_n^M$ is the initiation set that depends on macro-action observation histories, defined as $h_{n,t}^M=\{o_n^0, m_n^1, \cdots,o_n^{t-1}, m_n^t\}\in H_{n}^M$, $\beta_{n}^m\!:\!S\rightarrow[0,1]$ is a stochastic termination condition that depends on the underlying states, and $\pi_{n}^m\!:\!H_n\times M_n\rightarrow[0,1]$ is an option policy for macro-action $m$ ($H_n$ is the space of history of primitive-action and observation). Macro-actions are natural representations for robot or human operation for completing a task (e.g., navigating to a way point or placing an object on a robot).
\begin{comment}
Figure~\ref{fig:MAflow} illustrates the flow of MA execution for each agent in MacDec-POMDPs. The part highlighted in yellow corresponds to a macro-action controller, and the rest corresponds to a primitive-action controller. Note that the primitive actions are assumed to be executed by closed-loop controllers (e.g., a ROS navigation controller) which are given. Hence, whenever a primitive action fails, the controller is able take an error (deviation) as feedback and revise the trajectory in order to continue executing the current macro-action.  %\todo{Do we really do the replanning in the rightmost box? If so, we shoul make sure to talk about it clearly in the paper. }\mXX{is it clear now?}%\sXX{can this fig be made vertically smaller? it's using up a lot of space}
\begin{figure}[!ht]
\vskip -0.1in
\centering
	\includegraphics[scale=0.35]{./figures/simulationflow2.png}
	\caption{A flowchart for of decision making for each agent at each time step. The boxes highlighted in yellow correspond to macro-action execution, and the rest parts correspond to primitive-action execution. \sXX{seems like the bottom right boxes should have outgoing arrows to somewhere? CA: I'm not sure we really need this figure. It may make things more confusing rather than more clear.} }
    \label{fig:MAflow}
    \vskip -0.1in
\end{figure}
\end{comment}
MacDec-POMDPs can be thought of as decentralized partially observable \emph{semi-}Markov decision processes (Dec-POSMDPs) \cite{CA:RSS15,ShayeganICRA15}, because it is important to consider the amount of time that may pass before a macro-action is completed.
 The high level policy for each agent $\Psi_n$, can be defined for choosing macro-actions that depends on macro-action observation histories. Given a joint policy, the primitive action at each step is determined by the high-level policy that chooses the MA, and the MA policy that chooses the primitive action. The joint high level policies and macro-action policies can be evaluated as: $V^\Psi(s_0) = \mathbb{E}\big[\sum_{t=0}^{\infty}\gamma^tr(s_t, \vec{a}_t)|s_0, \pi,\Psi\big]$~\footnote{Note that MacDec-POMDPs allows asynchronous decision making, so synchronization issues must be dealt with by the solver as part of the optimization. Some temporal constraints (e.g., timeouts) can be encoded into the termination condition of a macro-action.}. %\mXX{Q: More discussion of how well macro actions allow scaling should be included. Also, are there any time synchronization issues with macro actions between agents. Can temporal constraints be used? A: Note that MacDec-POMDPs allows asynchronous decision making, so synchronization issues must be dealt with by the solver as part of the optimization. Some temporal constraints (e.g., timeouts) can be encoded into the termination condition of a macro-action.} In this paper, our goal is to optimize a high-level policy based solely on macro-actions and macro-action observations (i.e., the underlying Dec-POMDP is unknown, but the set of macro-actions and observations is given).  The macro-action observations correspond to semantic object labels and can be obtained through training a sequential classifier~\cite{Shayegan:macObs}. %\todo{The citation isn't complete for this paper. Has it been published already?}.
 
\subsection{Solution Representation} 
A Finite State Controller (\textbf{FSC}) is a compact way to represent a policy as a mapping from histories to actions. Formally, a stochastic FSC for agent $n$ is defined as a tuple $\Theta_n = \langle Q_{n},M_n,O_n,\delta_{n},\lambda_{n},\mu_{n}\rangle$, where, $Q_{n}$ is the set of nodes\footnote{A controller node can be understood as a decision state (summary of history). They are commonly used for policy representation when solving infinite horizon POMDPs~\cite{kaelbling1998planning} and Dec-POMDPs~\cite{Book16}.}; %\mXX{Q: The purpose of the controller nodes should be explained better; more intuitively. A: A controller node can be understand as a decision state (summary of history). They are commonly used for policy representation when solving infinite horizon POMDPs~\cite{kaelbling1998planning} and Dec-POMDPs~\cite{Book16}.}; 
$M_n$  and $O_n$ are the output and input alphabets (i.e., the macro-action chosen and the observation seen); $\delta_{n}:Q_{n}\times O_n \times Q_{n} \rightarrow [0, 1]$ is the node transition probability, i.e., $\delta_n(q,o,q') = \mathrm{Pr}(q'|q,o)$; $\lambda_{n}^{0}:Q_{n}\times M_n \rightarrow [0, 1]$ is the output probability for node $q_{n,0}$, such that $m_{n,0} \sim \lambda_{n}^{0}(q_{n,0},m_{n,0})=\mathrm{Pr}(m_{n,0}|q_{n,0})$; $\lambda_{n}:Q_{n}\times O_n \times M_n \rightarrow [0, 1]$ is the output probability for nodes $\ne q_{n,0}$ that associates output symbols with transitions, i.e. $m_{n,\tau}\sim\lambda_{n}(q_{n,\tau},o_{n,\tau}, m_{n,\tau}) = \mathrm{Pr}(m_{n,\tau}|q_{n,\tau},o_{n,\tau})$; $\mu_{}: Q_{n} \rightarrow[0, 1]$ is the initial node distribution $q_{n,0}\sim \mu_{n}=\mathrm{Pr}(q_{n,0})$. This type of FSC is called a Mealy machine~\cite{amato2010finite}, where an agent's local policy for action selection $\lambda_n(q, o, m)$ depends on both current controller node (an abstraction of history) and immediate observation. By conditioning action selections on immediate observations, a Mealy machine can use this observable information to help ensure a valid macro-action controller is constructed~\cite{Liu:AAAI16}. %We will discuss the validity issue in the next section.

\subsection{Policy Learning Through EM}
A Dec-POMDP problem can be transformed into an \textbf{inference problem} and then efficiently solved by an EM algorithm. Previous EM methods~\cite{kumar2015probabilistic,Song:AAAI16} have achieved success in scaling to larger problems, but these methods require a Dec-POMDP model both to construct a Bayes net and to evaluate policies. When the exact model parameters $T$, $\Omega$ and $R$ are unknown, a Reinforcement Learning (RL) problem must be solved instead. To this end, EM has been adapted to model-free RL settings to optimize FSCs for Dec-POMDPs~\cite{WU:IJCAI2013,Liu:IJCAI2015} and MacDec-POMDPs~\cite{Liu:AAAI16}. 

For both purposes of self-containment and ease of analyzing new algorithm, we first review the policy based EM algorithm (PoEM) developed for the MacDec-POMDP case~\cite{Liu:AAAI16}. 

%\subsection{Learning Objective}
\begin{definition} (Global empirical value function) \label{def:obj-value}Let $\mathcal{D}^{(K)}\!\!=\!\!\{(\vec{o}^k_{0}, \vec{m}^k_{0}, r^k_{0},\cdots\vec{o}^k_{T_k}, \vec{m}^k_{T_k}, r^k_{ T_k})\}_{k=1}^{K}$ %\sXX{better notation here? seems to imply multiplication without commas} 
be a set of episodes resulting from $|N|$ agents who choose macro-actions according to $\Psi\!\!=\!\!\otimes_n\Psi_n$, a set of arbitrary stochastic policies with $p^{\Psi_{n}}(m|h)>0$, $\forall$ action $m$, $\forall$ history $h$. The global empirical value function is defined as 
\begin{equation}\label{eq:emprical_value}
\!\!\hat{V}\big(\mathcal{D}^{(K)};\Theta\big)\!\!\stackrel{def.}{=}\!\!\frac{1}{K}\sum_{k=1}^K\sum_{t=0}^{T_k}\gamma^tr_t^k\prod_{n=1}^{N}\frac{p(m^k_{n, 0:t}|h^k_{n,t},\Theta_n)}{p^{\Psi_n}(m^k_{n, 0:t}|h^k_{n,t})}
\end{equation} 
where $h_{n,t}^k\!=\!(m^k_{n,0:t-1},o^k_{n,1:t})$, $0\!\leq\!\gamma<1$ is the discount.% and $R_{min}$ is the minimum reward.
\end{definition}

%\begin{equation}\label{eq:V-full}
%\hat{V}\big(\mathcal{D}^{(K)};\Theta\big)\!\!\stackrel{def.}{=}\!\!\sum_{k=1}^K\sum_{t=0}^{T_k}\gamma^tr_t^k\prod_{n=1}^{N}\frac{p(m^k_{n, 0:t},q^k_{n, 0:t}|h^k_{n,t},\Theta_n)}{p(m^k_{n, 0:t},q^k_{n, 0:t}|h^k_{n,t},\Psi_n)}
%\end{equation} 
Definition~\ref{def:obj-value} provides an off-policy learning objective: given data $\mathcal{D}^{(K)}$ generated from a set of behavior policies $\Psi$, find a set of parameters $\Theta\!=\!\{\Theta_i\}_{i=1}^{|N|}$ such that $\hat{V}\big(\mathcal{D}^{(K)};\!\Theta\big)$ is maximized. Here, we assume a factorized policy representation $p(\vec{m}_{0:\tau}^k|\vec{h}_{1:\tau},\Theta) = \prod_{n=1}^{|N|} p(m^k_{n, \tau}|h^k_{n,\tau},\Theta_n)$ to accommodate decentralized policy execution.

\subsection{PoEM}
%The empirical value function given in Definition~\ref{def:obj-value} corresponds to a constrained off-policy optimization problem. However, 
Direct maximization of $\hat{V}\big(\mathcal{D}^{(K)};\Theta\big)$ is difficult; instead, $\hat{V}\big(\mathcal{D}^{(K)};\Theta\big)$ can be augmented with controller node sequences $\{\vec{q}_{0:t}^{\;k}:k=1\dots,K, t = 1 : T_k\}$ and maximize the lower bound of the logarithm of $\hat{V}\big(\mathcal{D}^{(K)};\Theta\big)$ (obtained by Jensen's inequality):
\begin{eqnarray}
	\label{eq:lb-ln-obj}
	&&\hspace{-1.25cm}\ln\hat{V}\big(\mathcal{D}^{(K)};\Theta\big)=\ln\mbox{$\sum\limits_{k,t,\vec{q}_{0:t}^{\;k}}\!\!\!
	\frac{f_t^{k}(\vec{q}_{0:t}^{\;k}|\widetilde{\Theta})\tilde{r}_t^kp(\vec{m}_{0:t}^k,\vec{q}_{0:t}^{\;k}|\vec{o}_{1:t}^{\;k},\Theta)}{f_t^{k}(\vec{q}_{0:t}^{\;k}|\widetilde{\Theta})}$}
		\cr
	&&\hspace{-1.25cm}\geq\mbox{$\sum\limits_{k,t,\vec{q}_{0:t}^{\;k}}\!\!\!
		 f_t^k(\vec{q}_{0:t}^{\;k}|\widetilde{\Theta})\ln\frac{\tilde{r}_t^{\;k}p(\vec{m}_{0:t}^k,\vec{q}_{0:t}^{\;k}|\vec{o}_{1:t}^{\;k},\Theta)}{f_t^k(\vec{q}_{0:t}^{\;k}|\widetilde{\Theta})}$}
		\stackrel{def.}{=}\!\mathrm{lb}(\Theta|\widetilde{\Theta}),
\end{eqnarray}
where $f_t^k(\vec{q}_{0:t}^{\;k}|\widetilde{\Theta})\!\!\stackrel{def.}{=}\!\!\tilde{r}_t^kp(\vec{m}_{0:t}^k, \vec{q}_{0:t}^{\;k}|\vec{o}_{1:t}^{\;k},\widetilde{\Theta})/\hat{V}(\mathcal{D}^{(K)};\widetilde{\Theta})$, and $\{f(\vec{q}_{0:t}^{\;k}|\widetilde{\Theta})\!\!\geq\!\!0\}$ satisfy the normalization constraint $\sum_{k=1}^K\sum_{t=0}^{T_k}\sum_{\vec{q}_{0:t}^{\;k}}f_t^k(\vec{q}_{0:t}^{\;k}|\widetilde{\Theta}) \!=\!K$ with $\widetilde{\Theta}$ the most recent estimate of $\Theta$, and  $\tilde{r}_t^k\stackrel{def.}{=}\gamma^t(r_t^k-r_{min})/\prod_{\tau=0}^tp^\Psi(\vec{m}^k_{\tau}|h^k_{\tau}), \forall t,k$ are reweighted rewards with $r_{min}$ denoting the minimum reward, leading to the following constrained optimization problem
\begin{eqnarray}\label{eq:EMproblem}
%\resizebox{1.0\hsize}{!}{$
&&\hspace{-0.5cm}\mbox{$\max_{\big\{f_t^k\big(\vec{q}_{0:t}^{\,k};\widetilde{\Theta}\big)\big\}, \Theta}\mathrm{lb}(\Theta|\widetilde{\Theta})$}
\cr
&&\hspace{-0.5cm}\mbox{$\textrm{subject to:} \sum_{k=1}^K\sum_{t=0}^{T_k}\sum_{q_{n,0:t}^{\,k}=1}^{|Q_{1:|N|}|} f_t^k(\vec{q}_{0:t}^{\,k};\widetilde{\Theta}) = K$,} %,\,\,\,\, f_t^k(\vec{q}_{0:t}^{\,k};\widetilde{\Theta}) \geq 0, \forall t, k,
\cr
&&\hspace{-0.5cm}
\mbox{$p(\vec{m}_{0:t}^{\,k}\vec{q}_{0:t}^{\,k};\widetilde{\Theta}) = \prod_{n=1}^{|N|} p(m_{n, 0:t}^{\,k}, q_{n, 0:t}^{\,k}|o_{n, 0:t}^{\,k},\widetilde{\Theta}_n).$}  
\end{eqnarray}

Based on the problem formulation~\eqref{eq:EMproblem}, an EM algorithm can be derived to learn the macro-action FSCs. Algorithmically, the main steps involve alternating between computing the lower bound of the log empirical value function~\eqref{eq:lb-ln-obj} (E-step) and parameter estimation (M-step). This optimization algorithm is called policy based expectation maximization (PoEM), the details of which is referred to~\cite{Liu:AAAI16}.

\section{Related work}
%\todo{I don't think we want to just have related work on SAR or else people will think the paper is primarily about solving the SAR problem. If you don't want to mention other topics (on multi-robot and multi-robot stuff more generally), then you could just move some of this text to the intro and ditch the rest.}
%\sXX{might it be better to have related works section earlier, to emphasize the gaps the proposed approach addresses early on?}\mXX{How about here? or we should move it to intro?}\sXX{I'd say right after intro or background}
The use of multi-robot teams has recently become viable for large-scale operations due to ever-decreasing cost and increasing accessibility of robotics platforms, allowing robots to replace humans in team-based decision-making settings including, but not limited to, search and rescue \cite{grayson2014search}. Use of multiple robots allows dissemination of heterogeneous capabilities across the team, increasing fault-tolerance and decreasing risk associated with losing or damaging a single all-encompassing vehicle \cite{wong2011multiple}.

%\cite{jin2003cooperative}

The large body of work on multi-robot task allocation (MRTA) comes in decentralized, centralized, and distributed/hybrid flavors. Centralized architectures \cite{jin2003cooperative,turra2004fast} rely on full information sharing between all robots. However, in settings such as SAR, communication infrastructure may be unavailable, requiring the use of alternative frameworks. Distributed frameworks, such as those used in auction-based algorithms \cite{choi2009consensus}, use local communication for consensus on robot policies. This enables robustness against communication failures in hazardous, real-world settings. However, in settings such as SAR, it can be unreasonable or impossible for robots to communicate with one another during task execution. Decentralized frameworks, such as Dec-POMDPs \cite{Bernstein:MOR-2002} and the approach proposed in this paper, target this setting, allowing a spectrum of policies ranging from communication-free to explicitly communication-enabled. The flexibility offered by decentralized planners makes them suitable candidates for multi-robot operation in hazardous or uncertain domains, such as SAR.

Finally, note that unlike the majority of the existing MRTA literature, the work presented here exploits the strengths of the MacDec-POMDP framework \cite{AmatoICRA15} to develop a unifying framework which considers sources of uncertainty, task-level learning and planning, temporal constraints, and non-deterministic action durations.

%The  goal  of  this  paper  is  to  present an efficient policy learning method for solving decentralized cooperative partially observable robotics problems and demonstrate application of this method on domains with a team of heterogeneous robots.
 %\sXX{I think a stronger contribution claim is needed here to separate this work from the previous PoEM paper. Maybe something about convergence issues, dependence on initial parameters, or computational efficiency addressed by the iterative MC sampling approach.}
%\todo{CA: Yes, we should claim both an improved algorithm and experiments showing that this is efficient and works well on the robots.}

\section{Iterative Sampling Based Expectation Maximization Algorithm}
%In previous work, the search algorithm operated in batch mode, and there is no model selection, parameters space are fully explored, hence it is might converge to suboptimal \sXX{sentence is somewhat confusing. "Model selection" as in FSC parameter selection/hyperparameter optimization? Not clear what it's referring to}. In this paper, we use a parallel batch of search algorithms with multiple initial values, and instead of directly fitting to all available data, validation using a holdout set is conducted for model selection. The resulting algorithm outputs a policy with the highest off-policy evaluation, which addresses the stability issues of previous approaches.

The PoEM algorithm~\cite{Liu:AAAI16} is the first attempt to address policy learning for MacDec-POMDPs with batch data. However, one of the biggest challenge for PoEM is that it only grantees convergence to a local solution, a problem often encountered when optimizing mixture models, such as the empirical value function~\eqref{def:obj-value}~\footnote{Note that the empirical value function~\eqref{def:obj-value} can be interpreted as a likelihood function for FSCs with the number of mixture components equal to the total number of subepisodes $\sum_{k=1}^KT_k$~\cite{Kumar:UAI10}.}. Moreover, PoEM is a deterministic algorithm for approximate optimization, meaning that it converges to the same stationary point if initialized repeatedly from the same starting value. Hence, PoEM  %\sXX{clarify deterministic here. deterministic initialization? or deterministic action selection?}\mXX{check if the previous sentence make sense}\sXX{maybe mention it is a deterministic algorithm with random initialization. Or a deterministic algorithm for approximate optimization.}%, which is prone to local optimality and 
%can be very sensitive to initial conditions for 
can be prone to poor local solution for more complicated real-world problems (as it will be shown in a later numerical experiment). To address these issues, we propose a concurrently (multi-threaded) randomized method called iterative sampling based Expectation Maximization (iSEM). The iSEM algorithm is designed to run multiple instances of PoEM with randomly initialized FSC parameters in parallel to minimize the probability of converging to a sub-optimal solution due to poor initialization. Furthermore, to exploit information and computational efforts on runs of PoEM which are clearly converging to poor values, iSEM allows re-sampling of parameters once convergence of $V(D_{test})$ is detected, increasing the chance of overcoming poor local optima. Because of the re-sampling step, which involves random reinitialization for threads converging to poor local value, iSEM can be deemed as a randomized version of the PoEM algorithm. This is essential for convergence to well-performing policies, since it widely known that global optimization paradigms are often based on the principal of stochasticity~\cite{horst2013handbook}.%\todo{CA: Clarify this and make it sound cooler (what's different from PoEM here and why is it better?).}

\begin{algorithm}[t]
	\captionsetup{font=normalsize}
	\caption{\textsc{iSEM}}\label{alg:iSEM}
	% \footnotesize
	\begin{algorithmic}[1]
		\Require Episodes $\mathcal{D}^{(K)}_{train}$,$\mathcal{D}^{(K)}_{eval}$, number of MC samples $M$, maximum iteration number $T_{max}$, threshold $\epsilon$, $J=\emptyset$, $\rm{Iter}=0$  
		\While {$I\neq\emptyset$ \textbf{or} $\rm{Iter}\leq T_{max}$}
                \State $I =\{1,\cdots,M\}\setminus J$, $\rm{Iter}=\rm{Iter}+1$ \label{lst:step1}
        \For {$i \in I$}
        \State Sample $\{\Theta_i\} \sim Dirichlet(1)$ \label{lst:step2_start}
		\State $\Theta_i^\infty$ = PoEM($\Theta_i$, $D_{train}$) 
        \State Compute $V(D_{eval}, \Theta_i^\infty)$ using (\ref{eq:emprical_value}) \label{lst:step2_end}
		\EndFor
        \State Compute $\Theta^* =\argmax_{i \in \{1,\cdots,M\}}V(D_{eval},\Theta_i^\infty)$
 \label{lst:step3_max}
 \State $J=\emptyset$\label{step4_start}
         \For {$i = 1 $ to $ M$}
         \If {$V(D_{eval},\Theta^*)-V(D_{eval},\Theta_i^\infty) < \epsilon$}
        \State $J = J\cup\{i\}$
        \EndIf\label{step4_end}
        \EndFor
        \EndWhile
		\State \Return Controller parameters $\Theta^*$.
	\end{algorithmic}
    \vskip -0.05in
\end{algorithm}
iSEM is outlined in Algorithm~\ref{alg:iSEM}. Domain experience data is first partitioned into training and evaluation sets, $\mathcal{D}^{(K)}_{train}$ and $\mathcal{D}^{(K)}_{eval}$. iSEM takes the partitioned data, the number of Monte Carlo samples (threads) $M$ and parameters controlling convergence as input, and maintains two sets, $I$ and $J$: $I$ records the indices of threads whose evaluation values are $\epsilon$ lower than the best value, and $J$ records the remaining thread indices (and is initialized as empty). iSEM iteratively applies four steps: 1) update $I$ (line~\ref{lst:step1}); 2) for the threads in $I$, randomly initialize FSC parameters by drawing samples from Dirichlet distributions with concentration parameter $1$, run the PoEM algorithm~\cite{Liu:AAAI16} and evaluate the resulting policy $\{\Theta_i^\infty\}_{i\in I}$\footnote{$\infty$ sign indicates run the PoEM algorithm until convergence.} (line~\ref{lst:step2_start}-\ref{lst:step2_end}); 3) update the best policy and its evaluation value obtained in current iteration (line 8); 4) update $J$ by recording the indices of threads whose converged policy values are $\epsilon$ close to the best policy (line~\ref{step4_start}-\ref{step4_end}). Critically, the final step (update of $J$) enables distinguishing threads that clearly converge to poor local solutions and "good" local solutions. In the subsequent iteration, threads with poor local solutions are reinitialized and re-executed until the policy values from all the threads are $\epsilon$ close to the best solution learned so far. %\sXX{any more discussion of set J to be included here? suggest focusing on the computation benefits a bit more, to highlight differences with poem}\mXX{how about the above discussion newly added?}\sXX{looks better, I also made some changes in wording} % the data into training set $$: 1) keep track of convergence of $V(D_{test})$ for each PoEM thread 2) once convergence is detected, compare to the top-N values from other threads/runs of PoEM. If current thread's value is much worse, then re-sample the parameters (either using naive approach of uniform sampling, or maybe Bayes update of Dirichlet(1) sampling distr using the parameter values from other threads) 
The  iSEM  algorithm  is  guaranteed  to  monotonically  increase  the  lower bound of empirical  value  function  over  successive  iterations and the convergence property is summarized by the following theorem. 
\begin{theorem}\label{th:iSEM}
Algorithm~\ref{alg:iSEM} monotonically increases $\hat{V}\big(\mathcal{D}^{(K)};\Theta\big)$, until convergence to a maximum.
\end{theorem}
\begin{proof}
%We can prove the above theorem by induction. 
Assume that $\Theta^*(t)$ is a policy with the highest evaluation value among the policies learned by all the threads at iteration $t$, and the set $J_t$ records the thread indices with corresponding policy value $\epsilon$ close to $V(D_{eval},\Theta^*(t))$. In the iteration $t+1$, the set $I_{t+1}$ contains the thread indices with corresponding policy values satisfy $V(D_{eval},\Theta^*(t))-V(D_{eval},\Theta_i(t))>\epsilon, \forall i\in I_{t+1} =\{1,\cdots,M\}\setminus J_t$. Starting from $t=0$, we have $V(D_{eval},\Theta^*(0))\geq V(D_{eval},\Theta_i^\infty(0)), \forall i\in I_0=\{1,\cdots,M\}$. In the next iteration (i.e., $t=1$), we have $|I_1|=|\{1,\cdots,M\}\setminus J_0|\leq|I_0|$. The steps~\ref{lst:step2_start}-\ref{lst:step2_end} allow the threads in $I_1$ to rerun with randomly reinitialized parameters. According to step~\ref{lst:step3_max} (Algorithm~\ref{alg:iSEM}), we can obtain $V(D_{eval},\Theta^*(1))\geq V(D_{eval},\Theta^*(0))$. Following the same analysis for $t>1$, we can obtain $V(D_{eval},\Theta^*(t))\geq V(D_{eval}, \Theta^*(t-1))$. Since $\{V(D_{eval},\Theta^*(t))\}_{t=0}^\infty$ is a monotone sequence and it is upper bounded by $\frac{R_{max}}{1-\gamma}$, according to Monotone convergence theorem, {$V(D_{eval},\Theta^*(t))$} has a finite limit, which completes the proof. 
 
  %\todo{CA: where is the proof? And say what is different about this convergence vs PoEM.}
\end{proof}
%The proof is trivial and hence omitted to save space. \sXX{I would include the proof, it seems we still have space}
%\todo[inline, color=green!40]{add proof if space allows}
%\begin{proof}
%A5ccording to the design of Algorithm~\ref{alg:iSEM},
%\end{proof}
Note that the convergence of iSEM is different from that of PoEM in the sense that iSEM updates a global parameter estimate based on feedbacks from several local optima (obtained from random initialization). It is also worth mentioning that with finite number of threads, iSEM might still converge to a local maximum. However, we can show that on average, iSEM has higher probability of convergence to better solutions than PoEM. Moreover, the iSEM algorithm can be considered a special case of evolutionary programming (EP)~\cite{simon2013evolutionary}, which maintains a population of solutions (i.e., the set of policy parameters in $J$). Yet, there are obvious differences between iSEM and PE. Notably, instead of mutating from existing solutions, iSEM resamples completely new initializations for parameters and optimizes them using PoEM. In additional, iSEM is highly parallelizable due to its use of concurrent threads.
\section{Experiments}
This section presents simulation and hardware experiments for evaluating the proposed policy learning algorithm. First a simulator for a large problem motivated by SAR is introduced. Then, the performance of iSEM is compared to previous work based on the simulated SAR problem. Finally, a multi-robot hardware implementation is presented to demonstrate a working real-world system.
%We test our algorithm in a search and rescue scenario using an UAV and three ground rovers \sXX{why is this problem interesting? what other papers have tackled this problem? why are their approaches not good?}. The results demonstrate that our algorithm allows the agents to save most and if not all of the victims given uncertainty in observations. \sXX{this sentence is a bit too general, do we have any concrete results on success rates?}

\subsection{Search and Rescue Problem}\label{app:sar}
The SAR problem involves a heterogeneous set of robots searching for victims and rescuing survivors after a disaster (e.g., bringing them to a location where medical attention can be provided). Each robot has to make decisions using information gathered from observations and limited communications with teammates. Robots must decide how to explore the environment and how to prioritize rescue operations for the various victims discovered.

The scenario begins after a natural disaster strikes the simulated world. The search and rescue domain considered is a 20 $\times$ 10 unit grid with $s=6$ designated sites: 1 muster site and 5 victim sites. All robots are initialized at the muster site. Victim sites are randomly populated with victims (6 victims total). Each victim has a randomly-initialized health state. While the locations of the sites are known, the number of victims and their health at each site is unknown to the robots. The maximum victim capacity of each site also varies based on the site size. Each victim's health degrades with time. 

An unmanned aerial vehicle (UAV) surveys the disaster from above. A set of 3 unmanned ground vehicles (UGVs) can search the space or retrieve victims and deliver them to the muster site, where medical attention is provided. The objective of the team is to maximize the number of victims returned to the muster site while they are still alive. This is a challenging domain due to its sequential decision-making nature, large size (4 agents), and both transition and observation process uncertainty, including stochasticity in communication. Moreover, as communication only happens within a limited radius, synchronization and sharing of global information are prohibited, making this a highly-realistic and challenging domain.

\subsection{Simulator Description}
\label{sec:sim}
%\mXX{Is there anything from ROS we can visualize here?}\sXX{Kavin, please add snapshots of RVIZ ROS simulator to the expts section}
%\todo{Is this really ROS yet? If so, maybe move the "primitive action" part to the end and say here what the robots are here.}
All simulation is conducted within the Robot Operating System (ROS) \cite{ROS2009}. The simulator executes a time-stepped model of the scenario, where scenario parameters define the map of the world, number of each type of robot, and locations and initial states of victims.

Each robot's macro-controller policy is executed by a lower-level controller which checks the initiation and termination conditions for the macro-action and generates sequences of primitive actions.

\subsubsection{Primitive Actions}
The simulator models primitive actions, each of which take one time-step to execute. The primitive actions for the robots include: (a) move vehicle, (b) pick-up victim (UGVs only), (c) drop-off victim (UGVs only) and (d) do nothing.
%\todo{CA: Give a little more detail about these (especially move). What exactly are these and are they stochastic?}
Observations and communication occur automatically whenever possible and do not take any additional time to execute.

Macro-action policies, built from these primitive actions, may take any arbitrary amount of time in multiples of the time-steps of the simulator. Macro-action durations are also non-deterministic, as they are a function of the scenario parameters, world state, and inter-robot interactions (e.g., collision avoidance).

\subsubsection{The World}
While the underlying robotics simulators utilized are three-dimensional, the world representation is in two dimensions. This allows increased computational efficiency while not detracting from policy fidelity, as the sites for ground vehicles are ultimately located on a 2D plane. The world is modeled as a 2D plane divided into an evenly-spaced grid within a rectangular boundary of arbitrary size. Each rescue site is a discrete rectangle of grid spaces of arbitrary size within the world.

%\mXX{The robot implementation is in 3D, are there any technique challenge worth mentioning?}\sXX{added a bit more detail on it above. kavin can add some if he has insight from the code?}

Some number of victims are initially located in each rescue site. Victim health is represented as a value from 0 to 1, with 1 being perfectly healthy and 0 being deceased. Each victim may start at any level of health, and its health degrades linearly with time.  If a victim is brought to the muster location, its health goes to 1 and no longer degrades. One victim at a time may be transported by a UGV to the muster, although this can be generalized to larger settings by allowing the vehicle to carry multiple victims simultaneously.%\mXX{We can generalize this by allowing taking multiple victims at a time}\sXX{added a line on this}

\subsubsection{Movement}
Simulated dynamical models are used to represent the motion of the air and ground vehicles within ROS. The vehicles can move within the rectangular boundaries of the world defined in the scenario.

UGV motion is modeled using a Dubins car model. Real-time multi-robot collision avoidance is conducted using the reciprocal velocity obstacles (RVO) formulation \cite{van2008reciprocal}. State estimates are obtained using a motion capture system, and processed within RVO to compute safe velocity trajectories for the vehicles. 
%\mXX{We may also want to discuss about some technical details such as collision avoidance} \kXX {I mentioned that RVO is handling collision avoidance. Do I need to talk about how it works?}\mXX{What is RVO? anything can be cited?}\sXX{I added some discussion of this above}

%Any obstacles in the world, including other ground vehicles, must be navigated around.  The ground vehicle uses a wall-following technique to get around obstacles.\mXX{This is not a good technique, how is it done for real robot? Need to mention the real implementation of low level control.} In the event that the ground vehicle is stuck, it will perturb its path by moving one step in a random direction.\mXX{We are using RVO now, right?}

UAV dynamics are modeled using a linearization of a quadrotor around hover state, as detailed in \cite{mellinger2012trajectory}. Since the UAV operates at a higher altitude than UGVs and obstacles, there are no restrictions to the air vehicle's movement.%\mXX{Is is true the the simulator in ROS as well?}\sXX{kavin, does the UAV also avoid obstacles?} \kXX{no, it does not, only UGVs avoid obstacles} 
These dynamics correspond to the transition model $T$ specified in the (Mac)Dec-POMDP frameworks discussed in the section~\ref{sec:model-framework}.
%\todo{CA: Relate this to the transition model in the (Mac)Dec-POMDP.}
%For the scenarios used to generate the results in this study, the ground vehicle can move exactly one grid space in one time-step of the simulator and the air vehicle moves approximately 5x faster.

\subsubsection{Communication}
Communication is range-limited.  When robots are within range (which is larger for UAV-UGV communication than for UGV-UGV communication), they will automatically share their observations with two-way communication.  Communication is imperfect, and has a .05 probability of failing to occur even when robots are in range.
%\todo{CA: But what does it communicate?}
For the scenarios used to generate the results in this study, a UGV can communicate its observation with any other UGV within 3 grid spaces in any direction; the UAV can communicate with any UGV within 6 grid spaces in any direction.

\subsection{MacDec-POMDP Representation}
We now describe the MacDec-POMDP represention that is used for learning. Note that the reprentation in Section \ref{sec:sim} is not observable to the robots and is only used for constructing the simulator. 
%Each robot has a set of macro-actions that can be executed by the robot's decentralized policy at each decision epoch. The macro-action is executed by a lower-level controller of primitive actions. 
%The macro-actions for each ground robot are: \todo{Remove this? You say what the macro-actions are at the end of the section, so this is repetitive.}
%\begin{itemize}
%\item Go to muster and drop off victims,
%\item Pick up victim at current site, and
%\item Go to site $i \in \{1,\ldots,s\}$.
%\end{itemize}
%The macro-actions for each air robot are:
%\begin{itemize}
%\item Go to site $i \in \{1,\ldots,s\}$.
%\end{itemize}
\subsubsection{Rewards}
The joint reward is $+1$ for each victim brought back to muster alive and $-1$ for each victim who dies.

\subsubsection{Observations}
In the SAR domain, a UAV can observe victim locations when over a rescue site. However, victim health status is not observable by air. A UGV that is in a rescue site can observe all victims (location and health status) within that site. Robots are always able to observe their own location and whether they are holding a victim at a given moment. %Observations are imperfect, with a .05 probability of being inaccurate.\todo{CA: What does that mean? What observation(s) do you get when they are inaccurate?}

The observation vector $O$ on which the macro-controller makes decisions is a subset of the raw observations each robot may have accumulated through the execution of the prior macro-action.  The robots report the state of their current location and one other location (which could be directly observed or received via communication while completing the macro action). The second location reported is the most urgent state with the most recent new observation. If there are no new observations other than the robot's own location, the second location observation is equivalent to the self location.

%The disaster area is divided into a discrete set of rescue ``sites,'' conceptually areas in which victims may be located.  The map is represented as a 2-dimensional horizontal plane.\mXX{Refer to a picture in ROS visualization tool?.}

The observation vector is as follows,
%\begin{equation}
% O = 
%\begin{bmatrix}
%\text{self state} \\
%\text{self location} \\
%\text{location state} \\
%\text{second location} \\
%\text{second location state} \\
%\end{bmatrix} \,,
%\end{equation}
\begin{equation}
\begin{aligned}
O = [\text{self state}, \text{self location}, \text{location state},  \\
\text{second location}, \text{second location state}]
\end{aligned}
\end{equation}
where, self state $\in$ \{$1/0$= is/not holding victim\}, self location $\in$ \{site 1, site 2, ..., site s\}, location state $\in$ \{$0=$ no victims needing help, $1=$ victims needing help (not critical), $2=$ victims needing help (critical)\}, second location $\in$ \{site 1, site 2, ..., site s\}, and second location state $\in$ \{$0=$ no victims needing help, $1=$ victims needing help (not critical), $2=$ victims needing help (critical)\}.
%\end{itemize}
%\begin{itemize}
%\item self state $\in$ \{$1/0$= is/not holding victim\},
%\item self location $\in$ \{site 1, site 2, ..., site s\},
%\item location state $\in$ \{$0=$ no victims needing help, $1=$ victims needing help (not critical), $2=$ victims needing help (critical)\},
%\item second location $\in$ \{site 1, site 2, ..., site s\}, and
%\item second location state $\in$ \{$0=$ no victims needing help, $1=$ victims needing help (not critical), $2=$ victims needing help (critical)\}.
%\end{itemize}
There are $18s^2$ possible observation vectors, making the observation space substantially larger than previous macro-action based domains \cite{ShayeganICRA15,AmatoICRA15}.
%.\mXX{Is this a challenge for G-DICE or other planning based methods?}

\subsubsection{Macro-Actions}
%\todo{CA: UAVs also can do these except Pick up? Make it clear.}
The macro-actions utilized in this problem are as follows:
\begin{itemize}
\item \textbf{Go to Muster} (available to both UAV and UGV): Robot attempts to go to the muster point from anywhere else, but only if it is holding a live victim.  If a victim is on-board, victim will always disembark at the muster.%\mXX{Is battery level considered in ROS simulator?}\sXX{I'm not sure - kavin any idea?} \kXX{nope, battery level not considered in ROS simulator}
\item \textbf{Pick up Victim} (available only to UGV): Robot (UGV only) attempts to go to a victim's location from a starting point within the site.  Terminates when the robot reaches the victim; also may terminate if there is no longer a victim needing help at the site (i.e., another robot picked the victim up first or the victim died).  If victim and robot are located in the same grid cell, the victim can be ``picked up''.
\item \textbf{Go to Site $i \in \{1,\ldots,s\}$} (available to both UAV and UGV): Robot goes to a specified disaster site $i$. Terminates when the robot is in the site.  Robot can receive observations of the victims at the site.
\end{itemize}

%\subsection{Macro-Controller}

%The macro-controller is the finite-state machine policy which specifies the transition to the next macro-action.  The macro-controller is completely described by the Action Matrix and the Transition Matrix.

%The action matrix, A, is defined as
%$$A\{i\}(j)=k$$
%where $i$ is the robot index, 
%$j$ is the node index,
%and $k$ is the macro action index\\

%The transition matrix, T, is defined as
%$$T\{i,j\}=[k \;\;\; o]$$
%where $i$ is the robot index, 
%$j$ is the current node index, 
%$k$ is the next node index, 
%and $o$ is an observation vector\\
\begin{figure*}[ht]
\centering
\hspace{-0.0cm}
\begin{subfigure}[t]{0.23\textwidth}
		\centering
			\includegraphics[scale=0.11]{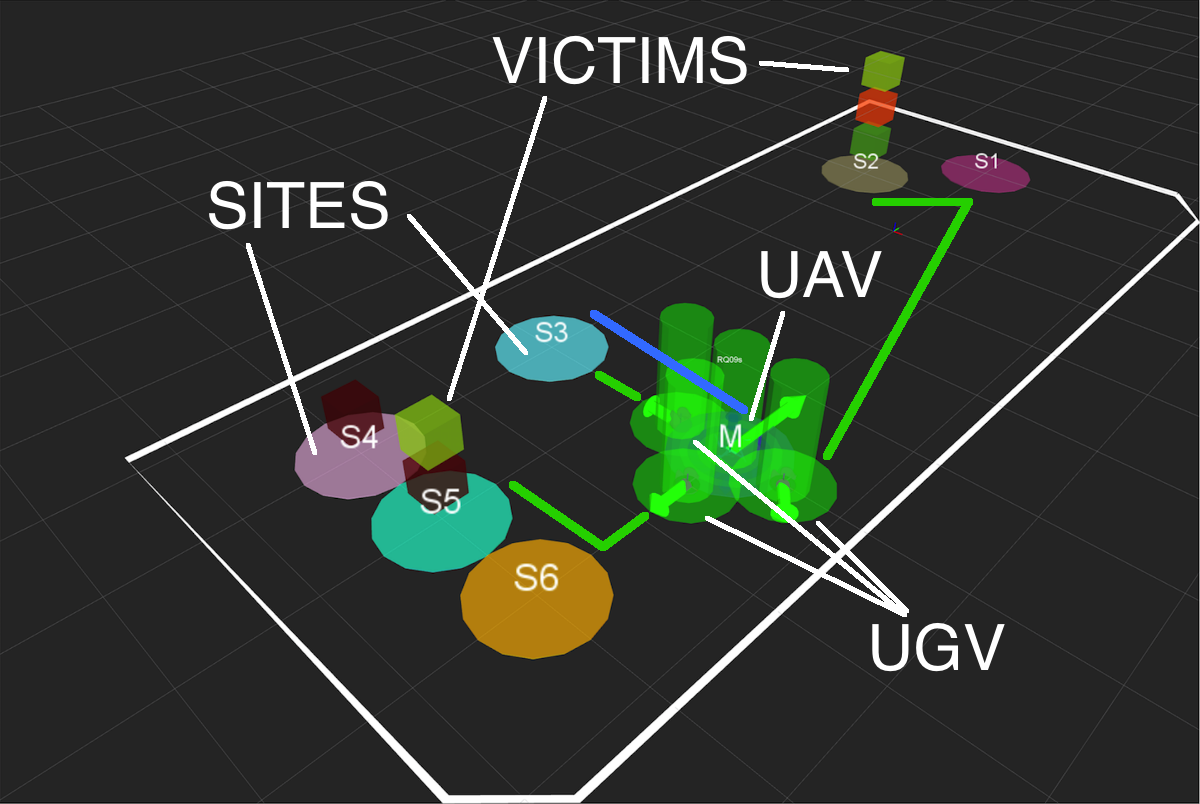}
		\caption{}\label{fig:victims_site}
	\end{subfigure}
    %\hfill
    \hspace{-0.02cm}
	\begin{subfigure}[t]{0.23\textwidth}
		\centering
			\includegraphics[scale=0.18]{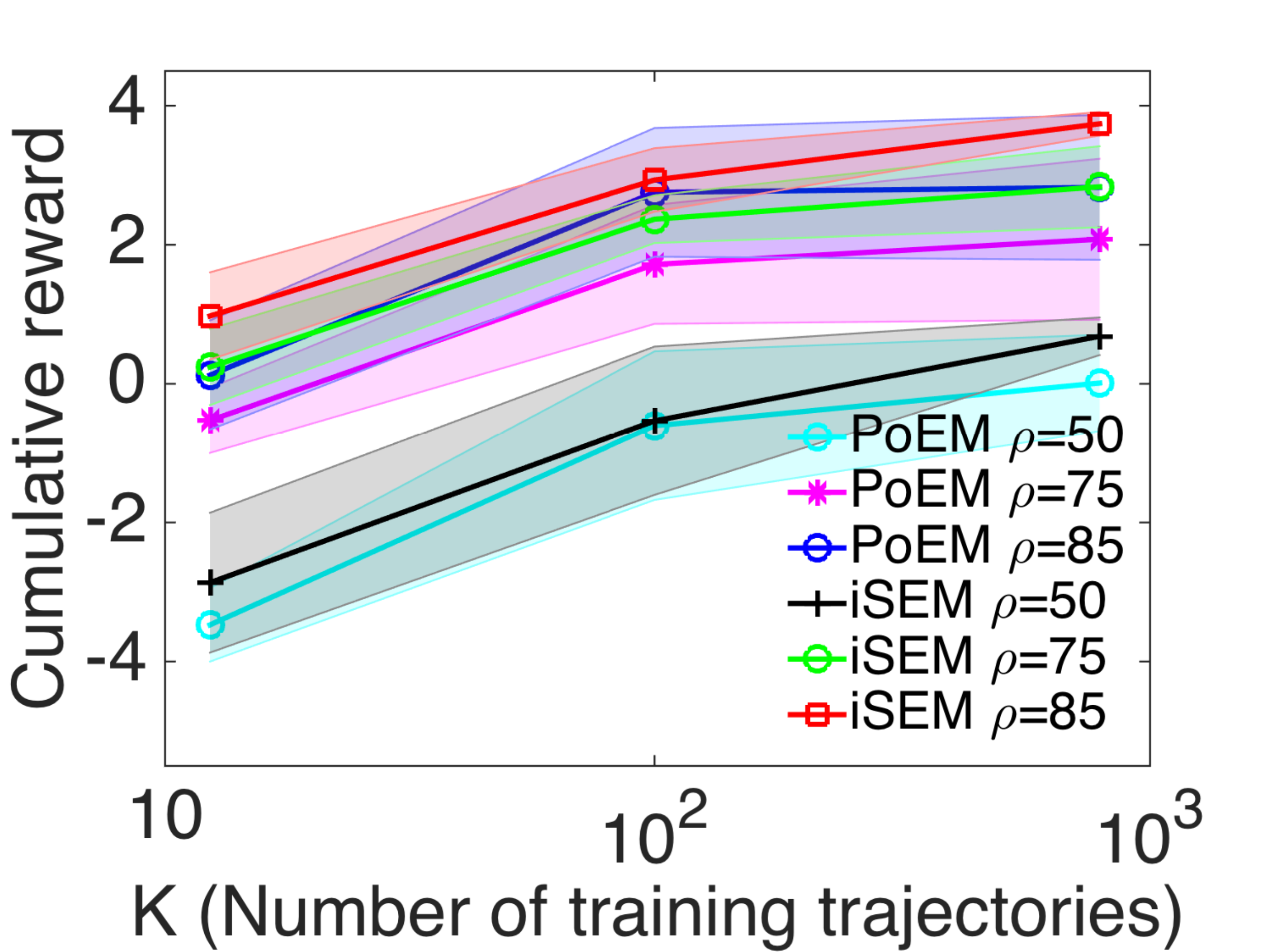}
		\caption{}\label{fig:iSEMvsPoEM}
	\end{subfigure}
	%\hfill
    \hspace{-0.02cm}
	\begin{subfigure}[t]{0.23\textwidth}
		\centering
		\includegraphics[scale=0.18]{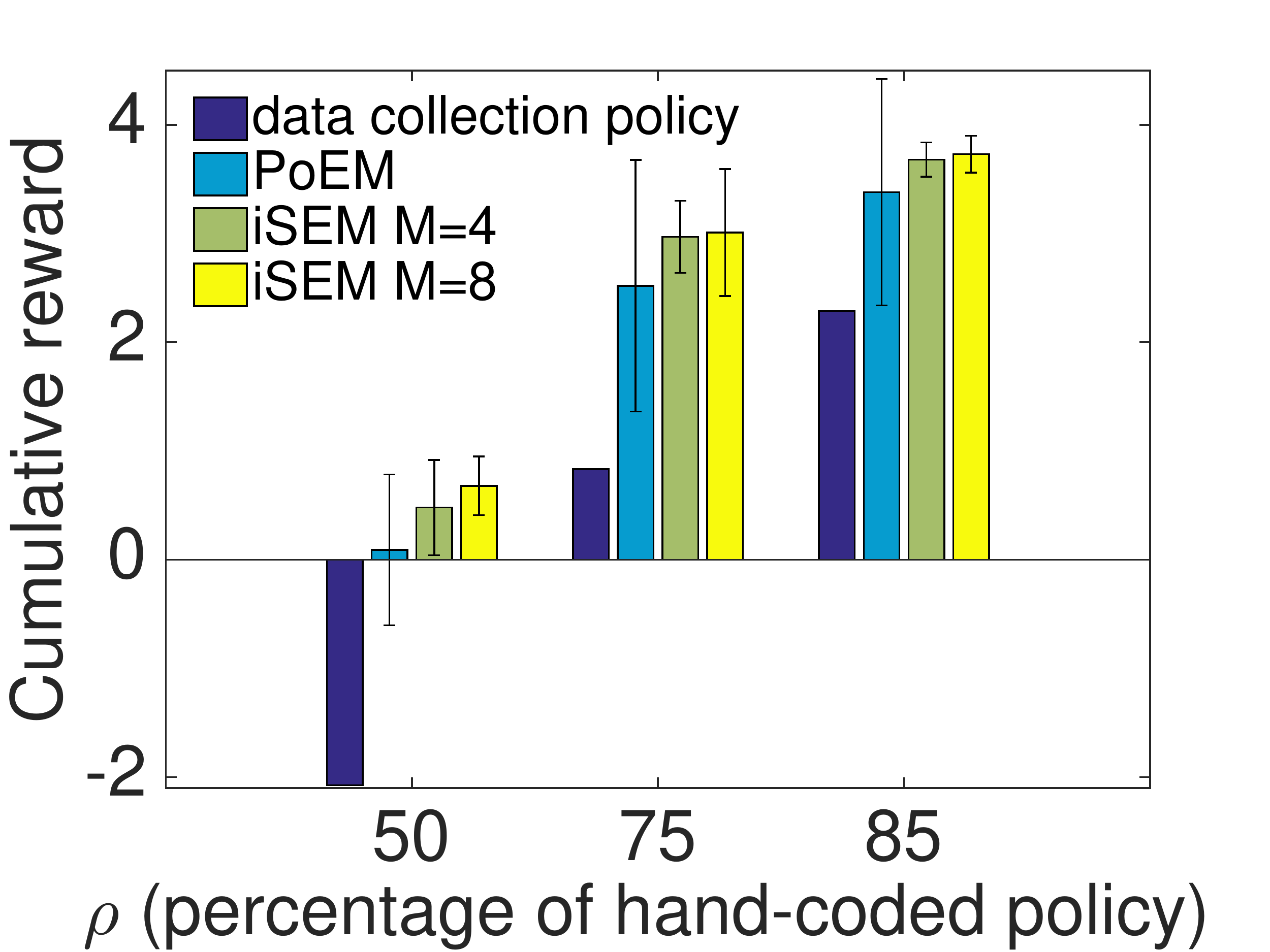}
		\caption{}\label{fig:M}
	\end{subfigure}
    %\hfill
    \hspace{-0.1cm}
	\begin{subfigure}[t]{0.23\textwidth}
		\centering
		\includegraphics[scale=0.18]{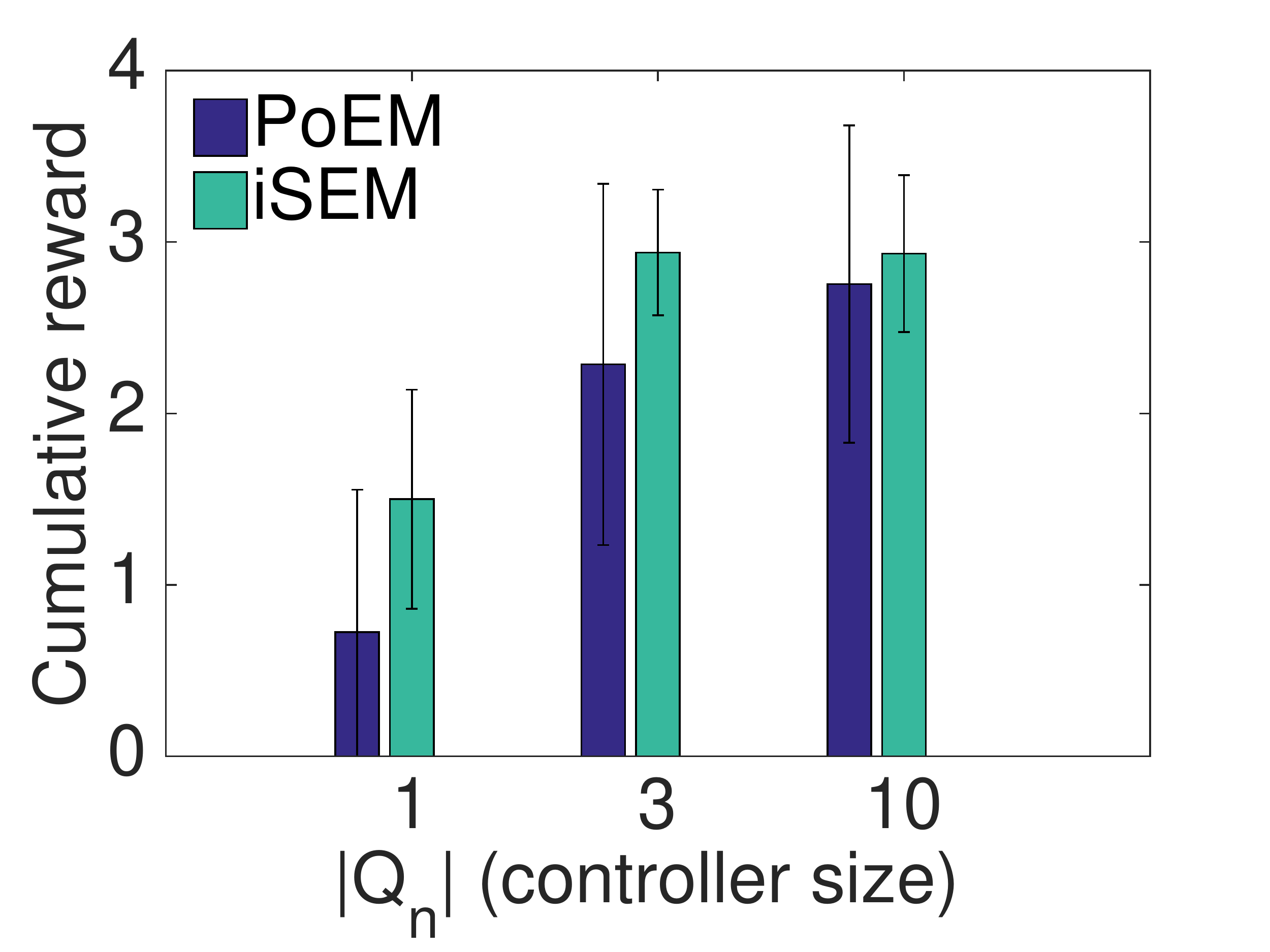}
		\caption{}\label{fig:Q}
	\end{subfigure}
    \vskip -0.1in
\caption{%\sXX{increase figure label sizes}
(a) RVIZ simulation of the experiment; Testing performance using (b) different number of training samples (with $Q_n$=10), (c) threads (with $Q_n$=10, $K$=500), and (d) controller sizes (with $M=8, K=100, \rho=85$).}
	\label{fig:results}
\end{figure*}
\subsection{Simulations and Numerical Results}
%\mXX{Do we need to mention that iSEM has also been tested on a $6\times6$ NAMO problem? The results of iSEM are very similar to PoEM. The reason might be due to the fact the domain/policy is too simple to have many local optima.}
The SAR domain extends previous benchmarks for MacDec-POMDPs both in terms of the number of robots and the number of states/actions/observations. Notably, due to the very large observation space cardinality of the SAR domain, it is difficult to generate an optimal solution with existing solvers such as \cite{ShayeganICRA15,AmatoICRA15} in a reasonable amount of time. Hence, due to the lack of a known global optima, the RL algorithms (iSEM and PoEM) are compared over the same datasets. The dataset is collected through the simulator by using a behavior policy combining a hand-coded expert policy (the same used in~\cite{Liu:AAAI16}) and a random policy, with $\rho$ denoting the percentage of expert policy.

To compare iSEM and PoEM on the SAR domain, experiments are conducted with $\rho=[50, 75, 85]$ and $|Q_n|=[1,3,10]$ (varying controller sizes)\footnote{$|Q_n|=1$ corresponds to reactive policies (based only on current observations).}. Corresponding test (holdout) set results are plotted in Figure~\ref{fig:results}. Several conclusions can be drawn from the results. First, as the amount of training data ($K$) increases, the cumulative reward increases for both PoEM and iSEM (under the same $\eta$, as shown in Fig.\ref{fig:iSEMvsPoEM}). Second, with the same $K$, iSEM achieves better performance than PoEM, which validates that iSEM is better at overcoming the local optimality limitation suffered by PoEM. In addition, as the number of threads $M$ increases, iSEM converges to higher average values and smaller variance (as indicated by the error-bar, compared to PoEM), according to Figure~\ref{fig:M}, which empirically justifies the discussion under Theorem~\ref{th:iSEM}. %\sXX{is figure 2d  discussed anywhere? also, I think there should be a lot more emphasis on results of fig 2c than on 2b, including some discussion on the smaller error bars for the iSEM case} %\sXX{strengthen these claims even more?}\mXX{will do after get more results.} 
Moreover, as shown in Fig.\ref{fig:Q}, under three settings of $|Q_n|$, the FSCs learned by iSEM render higher value than the PoEM policy. As $|Q_n|$ increases, the difference between PoEM and iSEM (with fixed $M$) tends to decrease, which indicates we should increase $M$ as iSEM is exploring higher dimensional parameter spaces. Finally, even in cases where the mean of iSEM is only slightly higher than PoEM, the variance of iSEM is is consistently lower than PoEM -- a critical performance difference given the uncertainty involved in the underlying domain tested.
%\todo{CA: Add more details here since the particle results and 2c and d are not clear. It currently sounds very underwhelming.} %\mXX{what is not clear to you, Chris? I know the results from $M=4$ and $M=8$ are kind of similar, but that is what I have get. I can include another set of results for $K=100$, in which the difference between $M=4$ and $M=8$ are more obvious.}

RVIZ \cite{rviz-ros} was used in conjunction with ROS to visualize the simulations. Fig.\ref{fig:victims_site} %\todo{CA:2a?} 
shows the start of one trial with the different colored circles being sites, the stacked cubes positioned at sites as victims with colors indicating their health values, and the 4 green cylinders indicating the 3 UGVs and the UAV. The sites are as follows (from furthest to closest): site 1 (red circle), site 2 (green), site 3 (sky blue), site 4 (pink), site 5 (turquoise), site 6 (orange). Note that the normal gridworld model used in the POMDP formulation usually assumes discrete state and discrete primitive actions, whereas the simulation models are based on macro-actions which comprised low-level controllers that can deal with both discrete and continuous primitive action and states.

%\todo{CA: make sure you talk about all the figures and the labels are right.}

\begin{comment}
\begin{figure}[!ht]
\centering
	\includegraphics[width=0.5\linewidth]{./figures/rviz_sim_final.png}
	\caption{The RVIZ simulation of the experiment}
    \label{fig:victims_site}
\end{figure}
\end{comment}

\subsection{Hardware Implementation} 

\begin{figure}[t]
	\begin{subfigure}[t]{0.21\textwidth}
		\centering
		\includegraphics[height=0.5\textwidth]{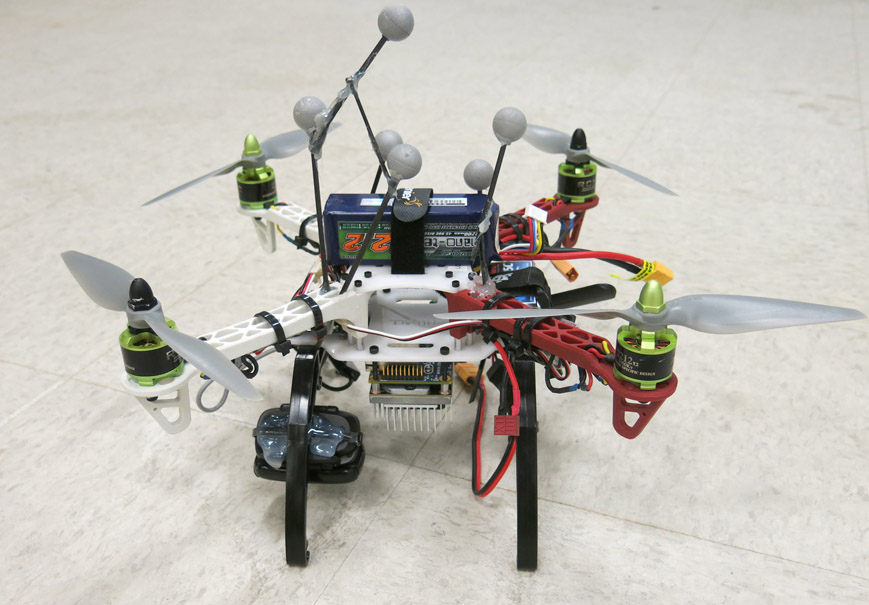}
		\caption{UAV is a DJI F330 quadrotor with onboard Jetson TX1 flight controller.}
		\label{fig:quad_overview}
	\end{subfigure}
	\hfill
	\begin{subfigure}[t]{0.21\textwidth}
		\centering
		\includegraphics[height=0.5\textwidth]{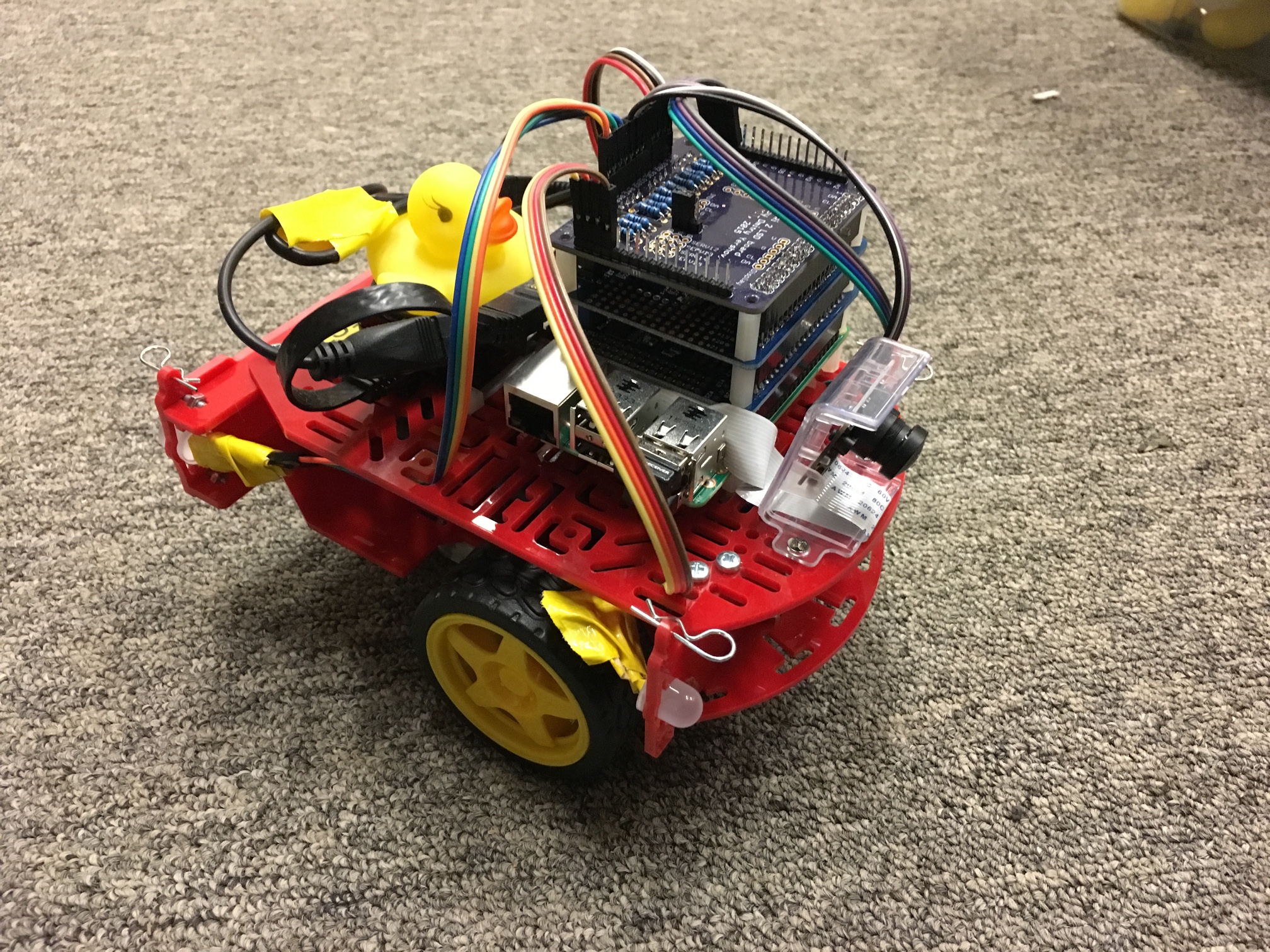}
		\caption{UGVs are custom-build ground robots with onboard Raspberry Pi 2.}
		\label{fig:ugv_overview}
	\end{subfigure}
   % \vskip -0.15in
    \caption{Robots in used Hardware Implementation.}
    \vskip -0.25in
\end{figure}

\begin{figure*}[t]
	\begin{subfigure}[t]{0.22\textwidth}
		\centering
		\includegraphics[width=0.9\textwidth]{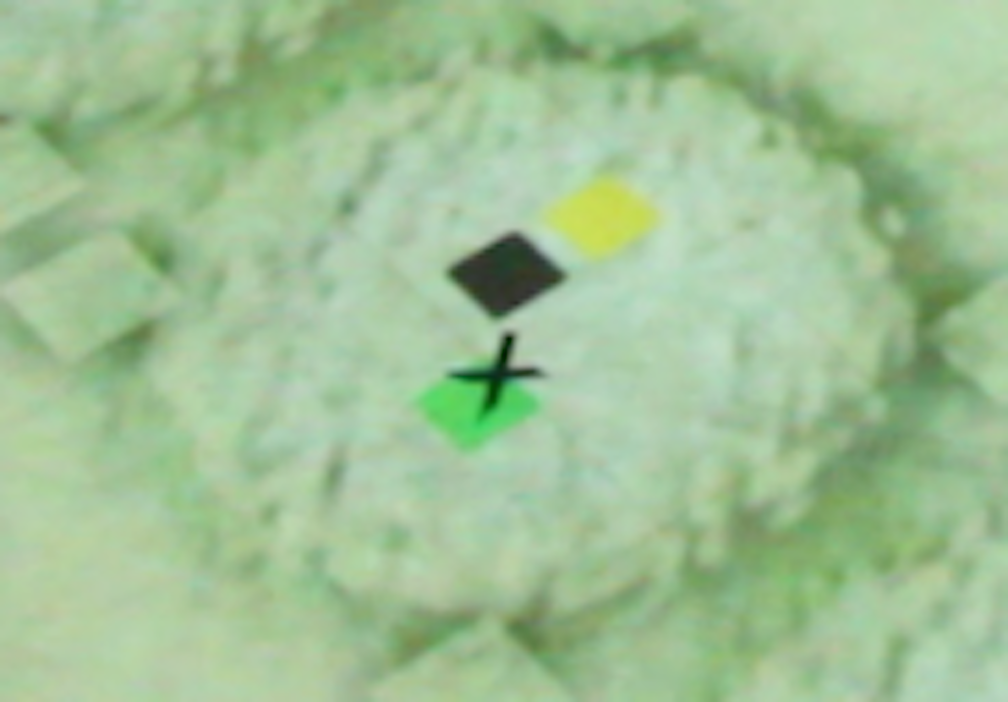}
		\caption{Zoomed in view of 3 victims (indicated as squares) at a particular site. }
		\label{fig:victims_at_site}
	\end{subfigure}
	\hfill
	\begin{subfigure}[t]{0.22\textwidth}
		\centering
		\includegraphics[width=0.9\textwidth]{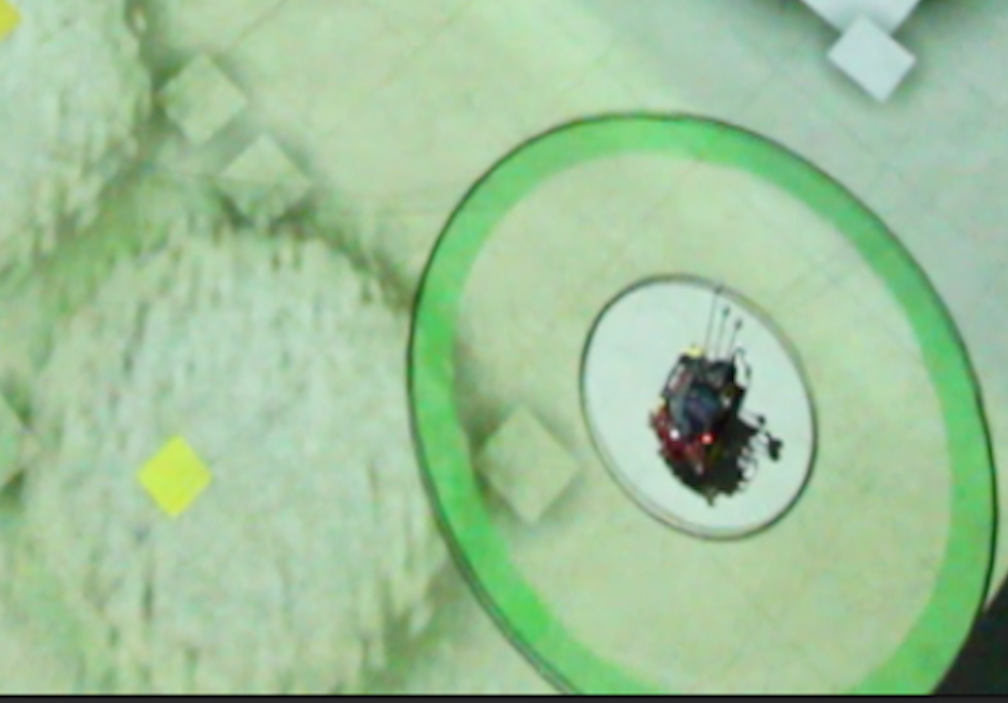}
		\caption{UGV observes the victim at the site with high health.}
		\label{fig:ugv_observation}
	\end{subfigure}
    \hfill
	\begin{subfigure}[t]{0.22\textwidth}
		\centering
		\includegraphics[width=0.9\textwidth]{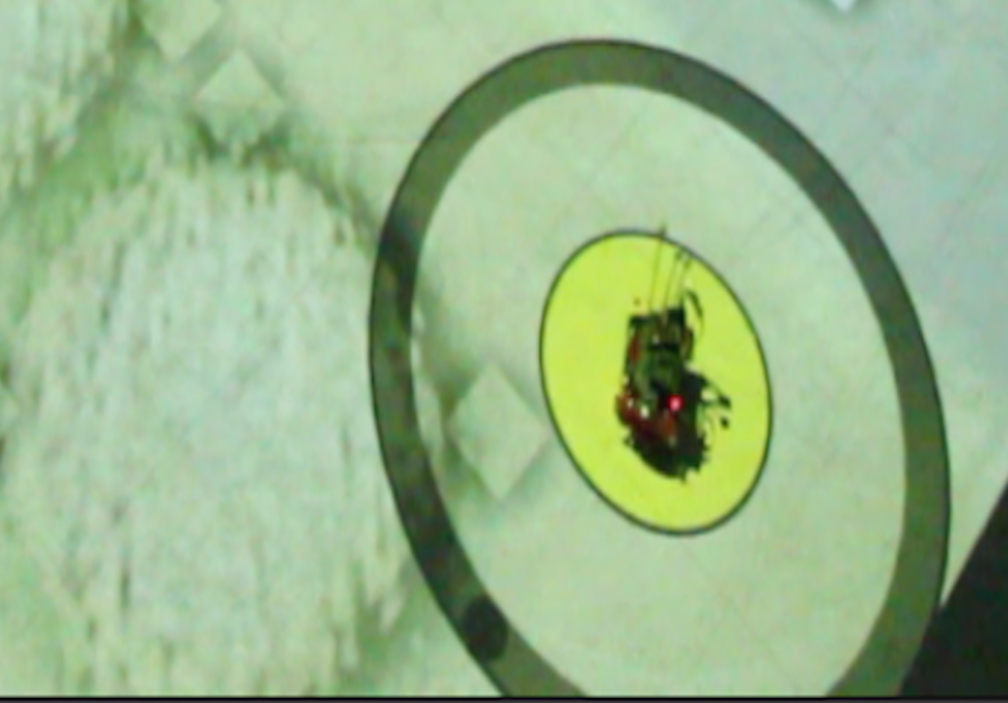}
		\caption{UGV outer ring color is black, as there are no other victims at this site.}
		\label{fig:ugv_pickingup}
	\end{subfigure}
    \hfill
	\begin{subfigure}[t]{0.22\textwidth}
		\centering
		\includegraphics[width=0.9\textwidth]{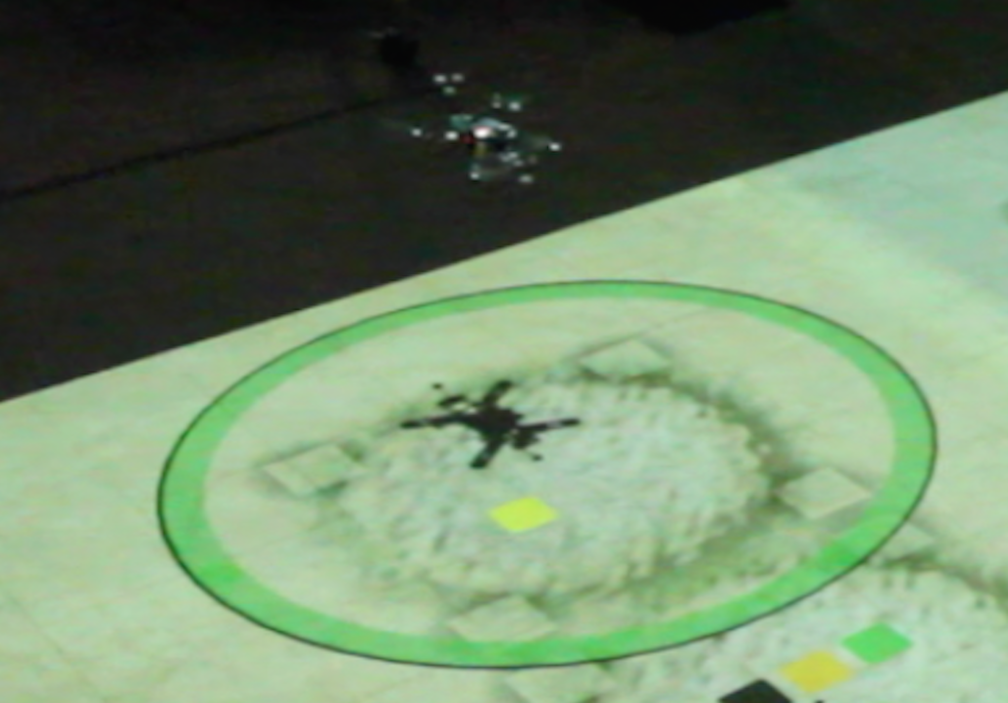}
		\caption{UAV can only observe, but not carry, victims. Thus, it only has an outer ring indicating observations.}
		\label{fig:uav_observation}
	\end{subfigure}  
    \vskip -0.1in
	\caption{Overview of hardware domain with 1 UAV and 3 UGVs. Projection system used to visualize sites and victim locations/health state. Victims shown as squares with colors representing health (green: high health, yellow: low health, red: critical health, black: deceased). For all robots, outer ring color indicates its noisy observation of the health of one of the victims present. For UGVs, inner circle color indicates health of the victim it is carrying.}
	\label{fig:overview_fig}
\end{figure*}

\begin{figure*}[t]
	\begin{subfigure}[t]{0.22\textwidth}
		\centering
		\includegraphics[width=0.9\textwidth]{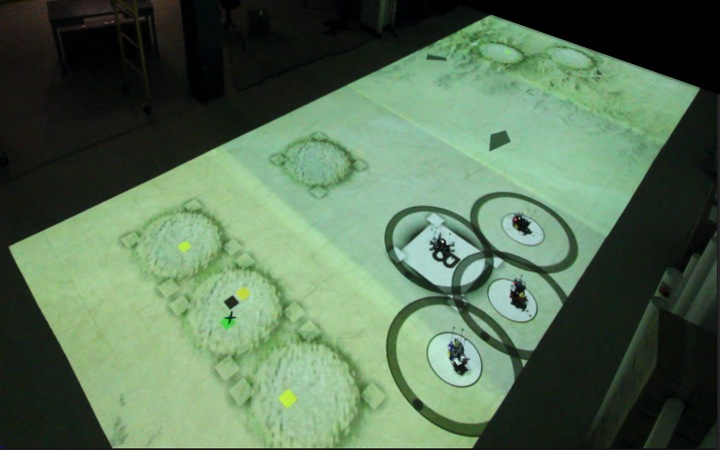}
		\caption{Start of experiment}
		\label{fig:seq1}
	\end{subfigure}
	\hfill
	\begin{subfigure}[t]{0.22\textwidth}
		\centering
		\includegraphics[width=0.9\textwidth]{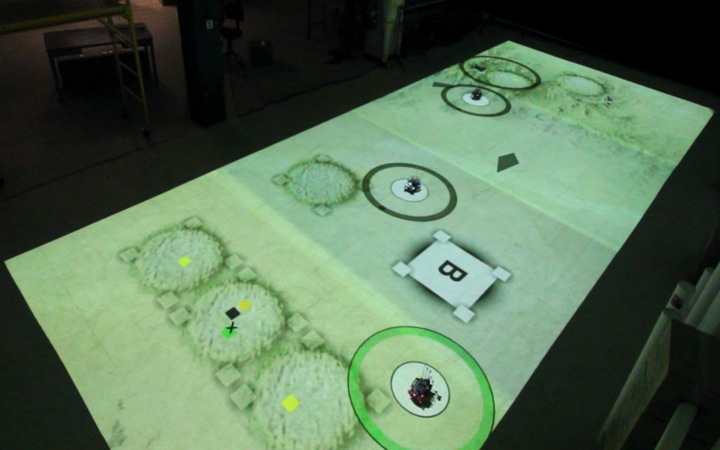}
		\caption{UGV observes victim at site 6}
		\label{fig:seq2}
	\end{subfigure}
    \hfill
	\begin{subfigure}[t]{0.22\textwidth}
		\centering
		\includegraphics[width=0.9\textwidth]{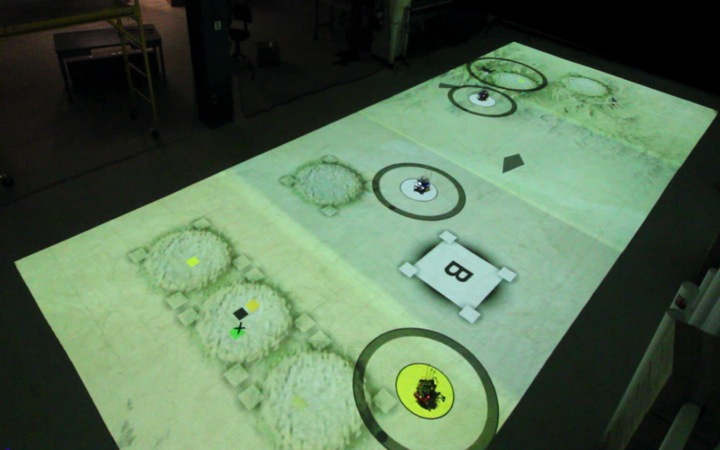}
		\caption{UGV picks up a victim, observes no others at site 6}
		\label{fig:seq3}
	\end{subfigure}
    \hfill
	\begin{subfigure}[t]{0.22\textwidth}
		\centering
		\includegraphics[width=0.9\textwidth]{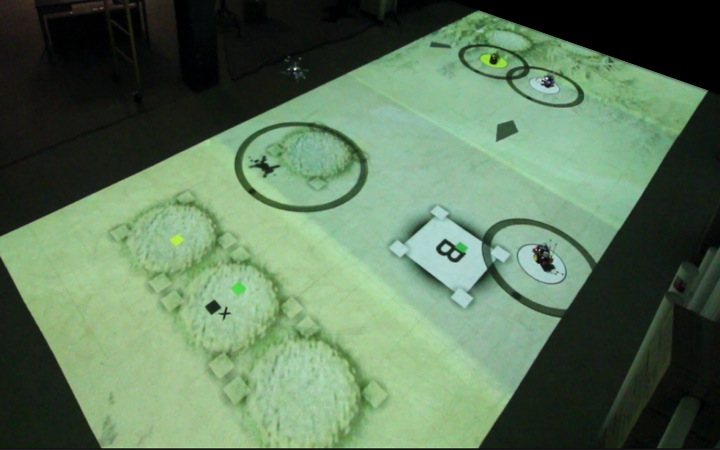}
		\caption{UGV drops off victim at muster}
		\label{fig:seq4}
	\end{subfigure}
	\\
	\begin{subfigure}[t]{0.22\textwidth}
		\centering
		\includegraphics[width=0.9\textwidth]{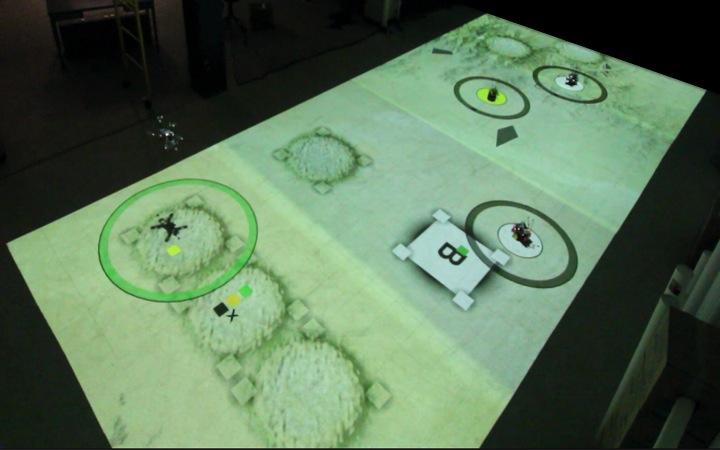}
		\caption{UAV observes victim at site}
		\label{fig:seq5}
	\end{subfigure}
	\hfill
	\begin{subfigure}[t]{0.22\textwidth}
		\centering
		\includegraphics[width=0.9\textwidth]{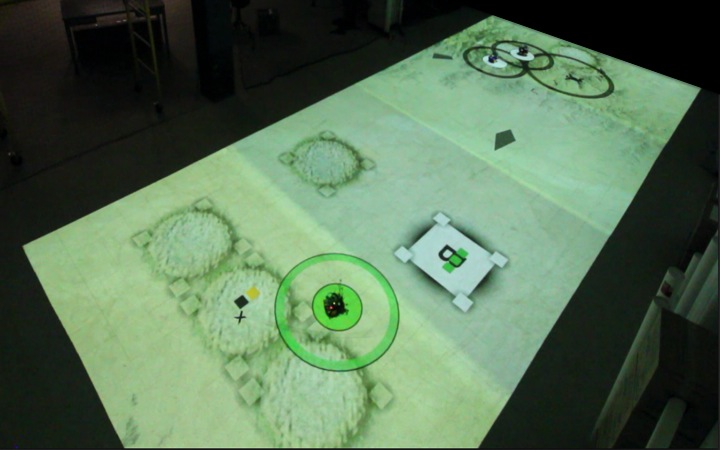}
		\caption{UGV picks up victim from site 5, observes another healthy victim at site}
		\label{fig:seq6}
	\end{subfigure}
    \hfill
	\begin{subfigure}[t]{0.22\textwidth}
		\centering
		\includegraphics[width=0.9\textwidth]{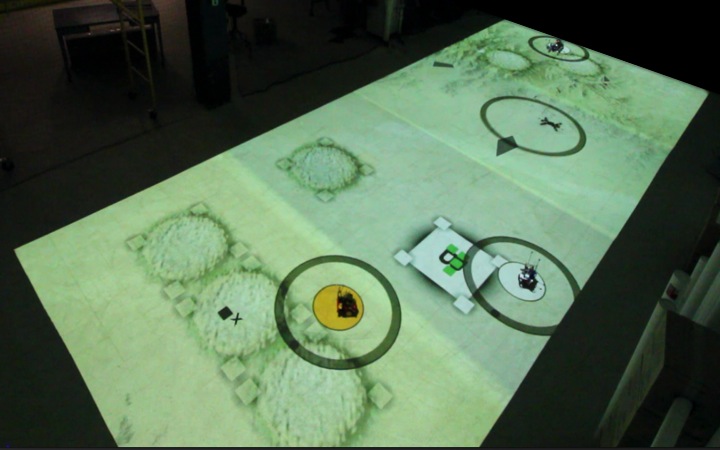}
		\caption{UGV picks up a victim from site 5, observes no more healthy victims at site 5}
		\label{fig:seq7}
	\end{subfigure}
    \hfill
	\begin{subfigure}[t]{0.22\textwidth}
		\centering
		\includegraphics[width=0.9\textwidth]{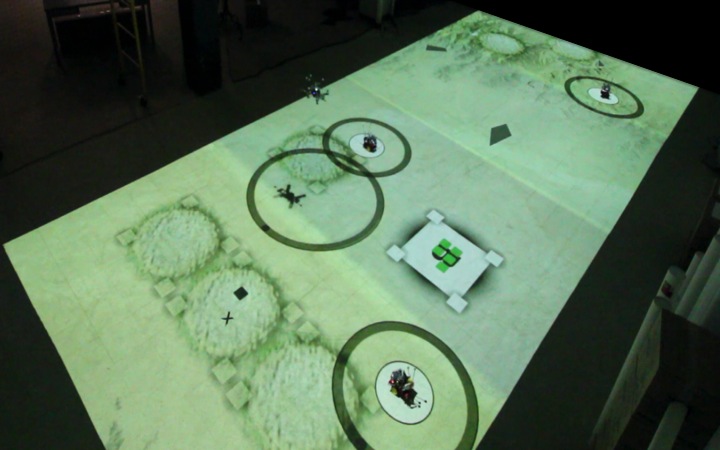}
		\caption{All healthy victims have been saved. End of experiment}
		\label{fig:seq8}
	\end{subfigure}    
	\vskip-0.1in 
	\caption{Overview of hardware domain, with 1 UAV and 3 UGVs. Ceiling-mounted projection system used to visualize sites and victim locations/health state.}
	\label{fig:snapshots_progression}
\vskip-0.3in
\end{figure*}

% \begin{figure}[!ht]
% \centering
% 	\includegraphics[scale=0.3]{./figures/domain_1.jpg}
% 	\caption{Overview of hardware domain, with 1 UAV and 3 UGVs. Ceiling-mounted projection system used to visualize sites and victim locations/health state. Colored rings indicate UAV and UGV positions, with sites indicated as textured circular regions.\sXX{crop out/put black box on ACL logo}}
%     \label{fig:domain_1}
% \end{figure}

%\mXX{show a picture with the rover and UAV, and describe their specifications. How is communication handled, what is the physical parameters of communication frequency, vehicle speed, the area of the lab? What ROS packages were used? More details regarding hardware implementation should be discussed. }

While simulation results validate that the proposed MacDec-POMDP search algorithm achieves better performance than state-of-the-art solvers, we also verify the approach on a SAR mission with real robots. This allows further learning from realworld experiences. A video demo is made available online\footnote{Video URL: \url{https://youtu.be/B3b60VqWMIE}}.  Learning from simulation allows robots operate in a reasonable (safe) way, whereas real robots experiments can potentially provide "realworld" experiences that are not fully captured by the simulators, hence allowing the robots to improve their baseline policy (learned from simulators). The video essentially demonstrates this potential, assuming the training data is collected from the ``realworld".
%\todo{CA: make it more clear why you need real robot results (still learning from the simulator?)}

A DJI F330 quadrotor is used as the UAV for hardware experiments, with a custom autopilot for low-level control and an NVIDIA Jetson TX1 for high-level planning and task allocation (Fig.~\ref{fig:quad_overview}). The UGVs are Duckiebots~\cite{paull2017duckietown}, which are custom-made ground robots with an onboard Raspberry Pi 2 for computation (Fig.~\ref{fig:ugv_overview}). Experiments were conducted in a 40 ft. $\times$ 20 ft. flight space with a ceiling-mounted projection system~\cite{omidshafiei2016measurable} used to visualize site locations, obstacles, and victims. As discussed earlier, limited communication occurs between robots, with a motion capture system used to ensure adherence to maximal inter-robot communication distances. 
    
%The maximum communication distance between two UGVs is half of the communication distance between a UGV and an UAV to take into account that the UAV is flying throughout the experiment, as height factors into the communication distance, thereby allowing UAV-UGV communication at normal horizontal distances. \sXX{unclear - how does cutting the distance in half take into account that the UAV is flying?}.

The hardware experiments conducted demonstrated that the policy generated from iSEM (with $\rho>75$, $K>100$, $Q_n>3$) was able to save all victims consistently well, despite robots having to adhere to collision avoidance constraints. In some instances, the robots were not able to save all 6 victims. However, in these scenarios, only 1 victim was lost, with the cause of loss due to an extremely low starting health for multiple victims. In such cases, an early victim death would occur before any robot could respond. 

Fig.~\ref{fig:snapshots_progression} shows the progression of one hardware trial. Sites are randomly populated with 6 victims total. All robots initiate at the muster site (Fig.~\ref{fig:seq1}). As the UGVs navigate towards sites (dictated by their policy), they simultaneously begin observing their surroundings. When they do, the outer ring surrounding them turns into the color of the latest victim observed (Fig.~\ref{fig:seq2}). A UGV can only pick up a new victim if it is not currently carrying a victim. Its inner circle then indicates the health of the victim it is carrying, while its outer ring indicates the health of a randomly-selected victim still present at the site (if any). Fig.~\ref{fig:seq3} illustrates a situation where no more victims are present at site 6, thereby causing the UGV's outer ring to turn black (no victims to save at latest encountered site). Note that an observed deceased victim also falls under this category. After a UGV picks up a victim, it drops it off at the muster (Fig.~\ref{fig:seq4}). The victim returns to full health, indicating a successful rescue. When a UAV visits a site, its outer ring also turns into the color of the victim observed at the site (Fig.~\ref{fig:seq5}). The UAV has no inner circle because it cannot pick up victims. Fig.~\ref{fig:seq6} and ~\ref{fig:seq7} show two more instances of a UGV picking up a victim from site 5. As mentioned before, a deceased victim results in a observation color of black in Fig.~\ref{fig:seq7}. Fig.~\ref{fig:seq8} shows the end of the hardware trial, where all healthy victims have been rescued.
%\todo{We need to more clearly explain why these results are good. Could any approach be used here and get the same results? Why not?}

\section{Conclusion}
This paper presents iSEM, an efficient algorithm which improves the state-of-the-art learning-based methods for coordinating multiple robots operating in partially observed environments. iSEM enables cooperative sequential decision making under uncertainty by modeling the problem as a MacDec-POMDP and using iterative sampling based Expectation Maximization trials to automatically learn macro-action FSCs. The proposed algorithm is demonstrated to address local convergence issues of the state-of-the-art macro-action based reinforcement learning approach, PoEM. Moreover, simulation results showed that iSEM is able to generate higher-quality solutions with fewer demonstrations than PoEM. %~\todo{CA: Say something like this "iSEM is able to generate higher-quality solutions with fewer demonstrations than the previous approach". In general, try to sell this a bit more.}
The iSEM policy is then applied to a hardware-based multi-robot search and rescue domain, where we demonstrate effective control of a team of distributed robots to cooperate in this partially observable stochastic environment. In the future, we will make our demonstration even closer to real world scenarios by modeling observations and communications as actions and assigning costs. We will also experiment with other methods other than random sampling, such as active sampling for the resampling step in iSEM, to accommodate restrictions of computational resources (i.e., number of threads). 
%We will also consider transfer learning settings 
 
%\clearpage
%\bibliographystyle{alpha}
%\bibliography{sample}
{	
	\bibliographystyle{IEEEtran}
    \small
	\bibliography{iros}
}
\end{document}